\documentclass[smallextended,envcountsame,numbook]{svjour3}
\smartqed

\usepackage{amsmath} 
\usepackage{amssymb}
\usepackage{graphicx,subfig} 
\usepackage{cite}
\usepackage{enumerate}
\usepackage[usenames]{color}
\usepackage{mathtools}
\usepackage{pstool}

\usepackage{bbm}

\usepackage{tikz}
\usetikzlibrary{decorations.pathreplacing}

  \usetikzlibrary{positioning}

\usepackage{pgfplots}

\usepackage{varioref}
\usepackage[breaklinks=true,%
            colorlinks=true,
            linkcolor=black,
            anchorcolor=black,
            citecolor=black,%
            filecolor=black,%
            menucolor=black,%
            pagecolor=black,%
            urlcolor=black,%
            plainpages=false,%
            bookmarksopen=true,
            bookmarksnumbered=true 
           ]{hyperref}

\usepackage{lineno,hyperref}
\usepackage{graphicx}
\usepackage{epstopdf}

\usepackage{float}

\begin{document}
\title{Data retrieval time for energy harvesting wireless sensor networks}

\author{
Mihaela Mitici \and Jasper Goseling \and Maurits de Graaf   \and  Richard J.\ Boucherie
}
\institute{M.~Mitici \and J.~Goseling  \and M.~de Graaf \and  R.J.~Boucherie \at
Stochastic Operations Research Group, Department of Applied Mathematics,\\
University of Twente,
P.O.~Box 217, 7500 AE, Enschede, The Netherlands. \\
\email{m.a.mitici@utwente.nl}}

\maketitle

%
%
%
\begin{abstract}
We consider the problem of retrieving a reliable estimate of an attribute monitored by a wireless sensor network, where the sensors harvest energy from the environment  independently, at random. Each sensor stores the harvested energy in batteries of limited capacity. Moreover, provided they have sufficient energy, the sensors broadcast their measurements in a decentralized fashion. Clients arrive at the sensor network according to a Poisson process and are interested in retrieving a fixed number of sensor measurements, based on which a reliable estimate is computed. 
We show that the time until an arbitrary sensor broadcasts has a phase-type distribution.
Based on this result and the theory of order statistics of phase-type distributions, we determine the probability distribution of the time needed for a client to retrieve a reliable estimate of an attribute monitored by the sensor network. We also provide closed-form expression for the retrieval time of a reliable estimate when the capacity of the sensor battery or the rate at which energy is harvested is asymptotically large. 
In addition, we analyze numerically the retrieval time of a reliable estimate for various sizes of the sensor network, maximum capacity of the sensor batteries and rate at which energy is harvested.  These results show that the energy harvesting rate and the broadcasting rate are the main parameters that influence the retrieval time of a reliable estimate, while deploying sensors with large batteries does not significantly reduce the retrieval time.

\end{abstract}

%
%
%

\section{Introduction} \label{sec:intro}

Managing energy consumption is an important component of wireless sensor network design, as it can lead to increased throughput and network lifetime. 
Recent technological advances have enabled sensor to harvest energy from the environment and to use this energy to recharge their batteries. Several technologies have been shown to be possible for energy harvesting such as solar energy, radio frequency, vibrations, thermoelectric (see, for instance, \cite{paradiso2005energy} for various examples of energy harvesting technologies).

In the case of battery-limited sensor networks, one of the main design goals is to minimize the sensor energy consumption so that the lifetime of the sensor network is extended. In contrast, energy harvesting techniques enable sensors to recharge their batteries over time, having the potential to extend the lifetime of the sensor network and to improve the overall performance of the network.
However, it is often the case that the harvested energy availability varies with time in a non-deterministic manner.  For instance, solar energy varies throughout the day and the intensity of the direct sunlight cannot be controlled. 
To enhance energy availability, sensor are often equipped with batteries where energy is stored for later use. The storage capacity of the batteries, however, is often limited. It is important, thus, to develop mechanisms to match the energy generation profile of the harvesting sensors with the  energy consumption
of the sensors. In particular, since sensors consume most of their energy to transmit their measurements \cite{akyildiz2002wireless}, it is important to jointly consider sensor transmissions and the process of energy harvesting.
Energy harvesting for wireless communications has resulted in a significant number of research publications in the last decade \cite{sudevalayam2011energy}. 
Considerable attention has been given to the coordination between sensor transmission scheduling, which is an energy-dependent process, and energy harvesting, which depends on the availability of the energy source. 
Various metrics for data transmission in energy  harvesting communication systems have been considered. 
In \cite{antepli2011optimal,yang2012broadcasting,yang2012optimal} the minimization of the time to transmit a fixed number of bits using a AWGN broadcast channel, where the transmitter harvests energy and has a finite-capacity rechargeable battery, is considered. In contrast to the model proposed in this paper, the energy arrival process is assumed to be known in advance, in an off-line manner. The authors optimize the transmission rate or the transmission power based on the energy available at the transmitter and on the amount of data that needs to be transmitted such that the transmission time is minimized. Structural properties of the transmission policies are derived.
Similar to the model proposed in this paper, in \cite{sharma2010optimal, tandon2014has} the process of energy harvesting is stochastic. Transmission policies that maximize the rate of data transmission to minimize the mean delay of data transmission are derived. In \cite{tandon2014has}, the average delay of data packets arriving according to a Poisson process at a transmitter which harvests energy, is derived. In comparison, in this paper we compute the expected delay of a set of sensor measurements to be transmitted by distinct sensors, each harvesting energy according to a Poisson process.

The problem of maximizing the amount of data transmitted up to a certain point is considered in \cite{devillers2012general,ozel2011transmission,ho2012optimal,tutuncuoglu2012optimum}. In \cite{devillers2012general} a general framework is provided to maximize the amount of transmitted data by a given deadline when the arrival process of energy is known in an off-line manner, at the transmitter and the battery of the transmitter is limited or suffers from energy leakage.  In  \cite{ozel2011transmission,ho2012optimal,tutuncuoglu2012optimum} dynamic programming is employed to determine an optimal energy allocation policy over a finite horizon so that the number of transmitted bits is maximized.
In \cite{lei2009generic} optimal transmission policies are derived to specify whether to transmit incoming data packets or to drop them based on a value attached to each packet and on the energy available at the transmitter.

Optimal transmission policies with energy harvesting nodes that transmit using fading wireless channels are considered in \cite{ozel2011transmission,ho2012optimal, kuan2014reliable}. In \cite{kuan2014reliable}, the probability of successful reception of data packets and the energy cost per transmitted packet are determined for energy harvesting devices that broadcast using non-perfect transmission channels. The authors propose an erasure-based broadcast scheme to guarantee reliable transmissions.

This paper considers the problem of retrieving a fixed number of sensor measurements over an  attribute  from distinct wireless sensors that harvest energy from the environment. 
Energy is harvested by each sensor according to a Poisson process, independently of the other sensors.
The fact that energy is harvested at random points in time reflects the stochastic nature of the availability of harvested energy.
We further assume that the batteries have limited storage capacity. When the battery of a sensor reaches the maximal storage capacity, additional harvested energy is discarded.
Provided that they have energy, the sensors broadcast measurements at an exponential rate and independently of each other.
Clients arrive at the network according to some Poisson process and are interested in retrieving several measurements over an attribute monitored by the sensors. Based on the retrieved sensor measurements, each client computes an estimate of the attribute. We impose that measurements are retrieved from distinct sensors to avoid the situation where the same measurement, from the same sensor, is retrieved multiple times, which could lead to a biased estimate.
 We determine the probability distribution of the retrieval time of a reliable estimate of an attribute monitored by the sensor network. Consequently, we provide a closed-form expression for the expected retrieval time of a reliable estimate. We also analyze the retrieval time of a reliable estimate when the capacity of the battery or the rate at which energy is harvested is asymptotically large. These results show the impact of the energy availability at the sensors, as well as the energy storage capabilities of the sensors, on the time required to retrieve a reliable estimate of an attribute from the sensor network.

The remainder of this paper is organized as follows. In Section \ref{ModelH} we formulate the model and the problem statement. In Section \ref{Analysis} we determine the distribution of the time for a client to retrieve a reliable estimate of an attribute from the sensor network. We also determine the retrieval time of a reliable estimate when the rate at which energy is harvested and the maximum capacity of the sensor batteries is asymptotically large. In Section \ref{Numerical} we compute numerically the retrieval time of a reliable estimate under various assumptions regarding the size of the sensor network, the energy harvesting rate and the maximum capacity of the sensor batteries.
In Section \ref{Conclusions} we discuss the results and provide conclusions.
%
%
%
\section{Model and Problem Statement} \label{ModelH}

We consider a network of $N$ wireless sensors, each having a noisy measurement on an attribute $\theta$. The measurements are subject to independent and identically distributed additive Gaussian noise with variance $\sigma^2$, i.e.\ $X_i \sim \mathcal{N}(\theta, \sigma^2)$.

Clients arrive at the sensor network according to a Poisson process with rate $\lambda_a$. Each client is  interested in acquiring from the network a sufficiently large set of sensor measurements based on which they can compute a reliable estimate $\overline{ X}$ of $\theta$. 

We consider an estimate to be reliable if the variance of the estimate $\overline{ X}$ at each client is below a targeted threshold $H$. Since 
\[ \textrm{Var} (\overline{X}) = \textrm{Var} (\frac{1}{s} \sum_{i=1}^{s} X_i)= \frac{1}{s^2} \sum_{i=1}^s \textrm{Var} (X_i)= \frac{\sigma^2}{s}<H,
\] 
we consider $\displaystyle s = \lceil \frac{\sigma^2}{H}\rceil$. Thus, any set of $s$ measurements from distinct sensors is sufficient for the client to retrieve a reliable estimate of $\theta$. As there are $N$ sensors which can provide $N$ distinct measurements, we also assume that $s \leq N$. 

Sensor $i$, $1 \leq i \leq N$, has $b_i$ units of energy, with $ 0 \leq b_i \leq B$. Each sensor harvests one unit of energy at an exponential rate $\lambda_e$, independently of the other sensors. The energy  is harvested from the environment. If $b_i=B$ and new energy is harvested, then this harvested energy is discarded.

We further assume that the sensor network broadcasts a measurement at an exponential rate $\mu$. This rate is shared uniformly by the $N$ sensors in the network. Thus, at an exponential rate $\mu/N$, independently of the other sensors, a random sensor broadcasts its measurement, provided it has energy. 

Upon a broadcast, the energy of the broadcasting sensor decreases by 1 unit (see also Figure \ref{fig:harvesting}). 

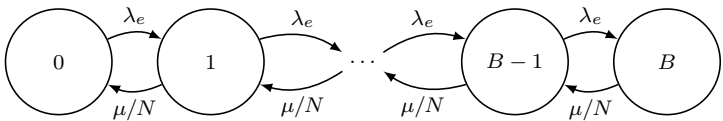
\begin{figure}[ht!]
\centering
\begin {tikzpicture}[-latex ,auto ,node distance =1 cm and 2cm ,on grid ,
semithick , state/.style ={ circle ,top color =white , draw , minimum width =1.4 cm}]
\node[state] (1)
{$1$};
\node[state] (0) [ left=of 1] {$0$};
\node (2) [ right =of 1] {$\ldots$};
\node[state] (3) [right =of 2] {$B-1$};
\node[state] (4) [right =of 3] {$B$};
\path (1) edge [bend left =25] node[below] {$\mu/N$} (0);
\path (0) edge [bend right = -25] node[above] {$\lambda_e$} (1);
\path (1) edge [bend left =25] node[above] {$\lambda_e$} (2);
\path (2) edge [bend right = -25] node[below] {$\mu/N $} (1);
\path (3) edge [bend left =25] node[above] {$\lambda_e$} (4);
\path (4) edge [bend right = -25] node[below] {$\mu/N $} (3);
\path (2) edge [bend left =25] node[above] {$\lambda_e$} (3);
\path (3) edge [bend right = -25] node[below] {$\mu/N $} (2);
\end{tikzpicture}
\caption{Birth-and-Death model for a single sensor that broadcasts, provided it has energy, at an exponential rate $\mu/N$ and harvests energy from the environment at an exponential rate $\lambda_e$. \label{fig:harvesting}}
\end{figure}

A measurement is retrieved by a client only if it is the first time this sensor broadcasts its measurement to this client. We call such a measurement to be  \textit{innovative} for this client. Thus, a client does not retrieve a measurement from the same sensor  multiple times.
Moreover, we assume that clients do not drop innovative measurements.

We are interested in the time $W_s$ for a client to retrieve a reliable estimate from the sensor network.

We end this section with some notation that will be useful for working with phase-type distributions. Let $\textbf{e}_k$ be a $k \times 1$ matrix with all unit entries.
Let $I_k$ denote the $k \times k$ identity matrix. For $n\times n$ matrix $M_1$ and $m \times m$ matrix $M_2$, let $M_1 \otimes M_2$ denote the Kronecker product of matrices $M_1$ and $M_2$ and let $M_1 \oplus M_2$ denote the Kronecker sum, i.e.
\[
M_1 \oplus M_2= M_1 \otimes I_m + I_n \otimes M_2.
\] 
Finally, let $M^{\otimes n}$ and $M^{\oplus n}$ denote the $n$-fold Kronecker product and the $n$-fold Kronecker sum with itself, respectively.

%
%
%
\section{Analysis}
\label{Analysis}

In this section we first determine the distribution of the time for a single sensor to broadcast, given that the system is in steady-state.  We show that this is a phase-type distribution. Using these results, we next determine the distribution of the time for a random client to retrieve
$s$ measurements from distinct sensors. Lastly, we compute the retrieval time of a reliable estimate $\overline{X}$ of $\theta$ for asymptotically large $B$, the maximum capacity of the sensor batteries, $\lambda_e$, the rate at which sensors harvest energy from the environment, and $N$, the size of the sensor network.

\subsection{A Single Sensor}

First, we consider  the steady-state probability that a random sensor has $i$ units of energy,  $0 \leq i \leq B$, denoted by $\nu(i)$,

The evolution of the units of energy at a sensor follows a Birth-and-Death model and a finite state space $\{0, 1, \dots, B\}$  with births at rate $\lambda_e$ and deaths at rate $\mu/N$ ( see Figure~\ref{fig:harvesting}). The steady-state distribution of such a model is well known in literature (see, for instance, \cite{asmussen2008applied}) and is, therefore, stated without proof below.
\begin{lemma}\label{pisensor}
The steady-state probability for an arbitrary sensor to have $i$ units of energy, $0 \leq i \leq B$, is:
\begin{equation}\label{pi0}
\nu(i)=
  \nu_0\left(\frac{\lambda_e N}{\mu} \right)^i, 
\end{equation}
where $\nu_0=(\lambda_eN/\mu-1)/((\lambda_e N/\mu)^{B+1} -1)$, if $\lambda_e \neq \mu/N$ and $\nu_0=1/(B+1)$ otherwise.
\end{lemma}
Note that in the above theorem $\nu_0$ is the probability that the battery of a sensor is depleted.

Next, we consider $W$, the time until a sensor broadcasts, given that the system is in steady-state. Based on $W$, we compute the time for a client to retrieve a reliable estimate by assuming that, upon arrival, this client observes the energy level of the sensors in steady-state. This is valid since the sensors operate independently of the arrivals of the clients. Moreover, since clients arrive according to a Poisson process, they \textit{see} the system in steady-state (PASTA).

Since the evolution of the energy at an arbitrary sensor follows a continuous-time Markov process, the distribution of $W$ can be  modeled as a phase-type distribution as follows.
Consider a continuous-time Markov chain with $B+2$ states. State $0 \leq i \leq B$ are transient states and correspond to a sensor having $i$ units of energy. The $(B+2)-$th state is an absorbing state. This state is reached when the sensor broadcasts a measurement. At an exponential rate $\lambda_e$, a jump occurs from state $i$ to state $i+1$, $0 \leq i <B$. This corresponds to the sensor harvesting an additional unit of energy.  At an exponential rate $\mu/N$, a transition occurs from state $1 \leq i \leq B$ to the absorbing state. This corresponds to a sensor broadcast. Let the initial distribution over the transient states be $\nu$. Then, the time until absorption is $W$, as desired.

Before giving a more formal description of the phase-type representation of $W$, we make a simplification by observing that in the above description the states $1$ to $B$ can be aggregated into a single transient state, which we will denote by $1$. There is a transition from state $0$ to this aggregated state $1$ at rate $\lambda_e$ and there is a single outgoing transition from this aggregated state $1$ to the absorbing state at rate $\mu/N$. Below, we will give the formal representation of this phase-type distribution as $(a, T)$ and specify the vector $a$ and the matrix $T$. Given this representation as a phase-type distribution, we immediately obtain $\mathbb{P}(W \leq t)=1-ae^{Tx}\textbf{e}_2$. In this case, however, since $T$ has such a simple structure we can also obtain the distribution function in an explicit form. This yields the following result.
\begin{lemma}\label{lemmaPW}
The distribution of $W$ is phase-type $(\textbf{a},T)$, where
\begin{equation}
\textbf{a} =
\begin{bmatrix}
\nu_0 \\
1 - \nu_0
\end{bmatrix},
\qquad
T =
\begin{bmatrix}
-\lambda & \lambda \\
0 & -\mu/N
\end{bmatrix}.
\end{equation}
The distribution function of $W$ can be expressed as
\begin{equation}
\mathbb{P}(W \leq t) = 1- e^{-\frac{\mu}{N} t}  + \frac{\frac{\mu}{N}}{\frac{\mu}{N} - \lambda_e} \nu(0) \left( e^{-\frac{\mu}{N} t} - e^{-\lambda_e t}\right).
\end{equation}
\end{lemma}
\begin{proof}
The phase-type characterization follows from the discussion above the lemma.
 The expression of the distribution function is obtained by observing that,
given that we are in state $0$, which happens with probability $\nu(0)$, the distribution of $W$ is given by the sum of two exponentially distributed random variables with parameters $\frac{\mu}{N}$ and $\lambda_{e}$ (see, for instance, \cite{akkouchi2008convolution}). Given that we are in the aggregated state $1$, which happens with probability $1-\nu(0)$, the distribution of $W$ is given by an exponentially distributed random variable with parameter $\frac{\mu}{N}$. 
Therefore,
\begin{equation*}
\mathbb{P}(W \leq t)=\left( 1-\frac{\frac{\mu}{N}}{ \frac{\mu}{N}- \lambda_e} e^{-\lambda_e t} + \frac{\lambda_e}{\frac{\mu}{N} - \lambda_e} e^{-\frac{\mu}{N} t} \right) \nu(0) + (1-e^{-\frac{\mu}{N} t}) (1-\nu(0)).
\end{equation*}
The result  follows directly from the above expression.
\end{proof}

\subsection{Retrieving a reliable estimate}

In this section, we determine the distribution of the time for a random client to retrieve $s$ measurements from distinct sensors. 

\begin{lemma}\label{lemmaPWs}
The distribution of the time until a client receives $s$ measurements from distinct sensors is:
\[
\mathbb{P}(W_s \leq t) =
		 1- \sum_{j=0}^{s-1} \binom{N}{j} \left( \sum_{k=0}^j \binom{j}{k} (-1)^{j-k}  \textbf{a}^{\otimes(N-K)}\exp\left(t T^{\oplus(N-k)}\right) \textbf{e}_{2^{N-k}} \right).
\] 
\end{lemma}

\begin{proof}
Recall that a client leaves the system as soon as he retrieves exactly $s$ measurements. Thus, we need to compute the distribution of the time  between the moment a client arrives at the network and the moment when  the $s$-th broadcast occurs, all $s$ broadcasts from  distinct sensors. This can be seen as the distribution of the $s$-th order statistic of $N$ phase-type distributed random variables with representation $(\textbf{a},T)$, as introduced above. The distribution of the $s$-th order statistic is (see, for instance, \cite{ahsanullah2013introduction}), for $N$ variables, is
\begin{equation}
\mathbb{P}(W_s \leq t)  = \sum_{j=s}^N \binom{N}{j} P(W \leq t)^j (P(W > t)^{N-j}.
\end{equation}
Starting from the above expression we derive
\begin{align}
\mathbb{P}(W_s \leq t)  	& =1-  \sum_{j=0}^{s-1} \binom{N}{j} P(W \leq t)^j (P(W > t)^{N-j} \notag \\
			&= 1-  \sum_{j=0}^{s-1} \binom{N}{j} \left(1- P(W>t) \right)^j  (P(W > t)^{N-j} \notag \\
			&= 1-  \sum_{j=0}^{s-1}  \hspace{-3pt} \binom{N}{j} \hspace{-3pt} \left( \sum_{k=0}^j \binom{j}{k} (-1)^{j-k} P(W>t)^{j-k} +1  \hspace{-3pt} \right)  \hspace{-3pt}  (P(W > t)^{N-j} \label{lineW1}\\
			&=1- \sum_{j=0}^{s-1} \binom{N}{j} \left( \sum_{k=0}^j \binom{j}{k} (-1)^{j-k} P(W>t)^{N-k} \right) \label{lineW2},
\end{align}
where in \eqref{lineW1} we expanded the polynomial $(1- P(W>t))^j$. 

Now, observe that the distribution of $P(W>t)^{N-k}$ in \eqref{lineW2} is:
\begin{align}\label{minW}
\mathbb{P}(W>t)^{N-k} = P(\min\{Y_1, Y_2, \ldots, Y_{N-k}\} >t), 
\end{align}
where the $Y_i, 1 \leq i \leq N-k$ are i.i.d. phase-type distributed random variables with representation $(a,T)$. Therefore, $P(W>t)^{N-k}$ is the first order statistic of  a phase-type distributed random variable for which it is well known (see, for instance, \cite{nielsen1988modelling}) that it is phase-type distributed with representation $(\textbf{a}^{\otimes(N-K)},T^{\oplus(N-k)})$. The result follows directly by inserting the distribution function of this phase-type distribution into~\eqref{lineW2}.
\end{proof}
The result above is general in the sense that it does not depend on the specific representation of the phase-type distributed random variable $W$.

We are next interested in determining $\mathbb{E}[W_s]$, the expected time for a client to retrieve exactly $s$ measurements. In principle, $\mathbb{E}[W_s]$ can be obtained directly from Lemma~\ref{lemmaPWs}. However, the moments of order statistics of phase-type distributed random variables are known in the literature~\cite{zhang2011computing}. Therefore, we will resort to the results from~\cite{zhang2011computing}. Let $m_s^k$ denote the $k$-th moment of the $s$-th order statistic of  $N$ phase-type distributed random variables with representation $(\textbf{a},T)$.
\begin{theorem}\cite[Thm 4.1]{zhang2011computing}\label{zhang}
\begin{align*}
m_s^k=m_{s-1}^{k} + \sum_{j=1}^s (-1)^{j-1} \binom{N-s+j}{j-1} L_{N-s+j}^{(k)},
\end{align*}
where $L_{j}^{(k)}= \binom{N}{j} (-1)^k k! \left( \textbf{a}^{\otimes j} \right) \left( T^{\oplus j} \right)^{-k} \textbf{e}_{2^{j}}$, $1 \leq j \leq s$, and $m_0^k=0$.
\end{theorem}
Taking $k=1$, $\mathbb{E}[W_s]$ can be computed from Theorem~\ref{zhang}. Note, however, that straightforward computation of moments based on Theorem~\ref{zhang} involves the matrices $T^{\otimes j}$, where $j$ takes values up to $N$. The dimension of $T^{\otimes N}$ is $2^N \times 2^N$. Therefore, the complexity of these computations is exponentially increasing in $N$. Since we are interested in the behaviour of the system for larger values of $N$,  we will derive in the next result an expression for $\mathbb{E}[W_s]$ that has at most polynomial complexity in all model parameters.
\begin{theorem}\label{firstmomentW}
The expected time, $\mathbb{E}[W_s]$, for a client to retrieve $s$ measurements from distinct sensors is:
\[
 \sum_{j=0}^{s-1} \binom{N}{j}   \sum_{k=0}^j \binom{j}{k} (-1)^{j-k}  \sum_{v=0}^{N-k} \binom{N-k}{v} \omega^v (1-\omega)^{N-k-v} \frac{1}{\lambda_e (N-k-v) +\frac{\mu}{N} v} , 
\] where $\omega=1-\nu(0) \frac{\frac{\mu}{N}}{\frac{\mu}{N} -\lambda_e}$.
\end{theorem}

\begin{proof}
The expected retrieval time for $s$ measurements from distinct sensors can be expressed using Theorem \ref{lemmaPWs} and Lemma \ref{lemmaPW} as follows.
\begin{align}
&\mathbb{E}[W_s] \notag \\
			&= \int_0^\infty \mathbb{P}(W_s>t) dt \notag \\
			&= \sum_{j=0}^{s-1} \binom{N}{j}  \int_0^\infty \left( \sum_{k=0}^{j} \binom{j}{k}  (-1)^{N-k} (1-\mathbb{P}(W \leq t))^{N-k}  \right) dt \label{w0}  \\
			&= \sum_{j=0}^{s-1} \binom{N}{j} \sum_{k=0}^{j} \binom{j}{k}  (-1)^{N-k} \int_0^\infty  \left( e^{-\frac{\mu}{N} t}\omega +e^{-\lambda_e t}(1-\omega)^{N-k} \right)^{N-k}  dt \label{w2}  \\
			&=  \sum_{j=0}^{s-1} \hspace{-3pt} \binom{N}{j} \hspace{-3pt} \sum_{k=0}^{j} \hspace{-3pt} \binom{j}{k}   \hspace{-3pt} (-1)^{N-k} 
			 \hspace{-4pt} \int_0^\infty \hspace{-2pt}  \sum_{v=0}^{N-k} \hspace{-3pt} \binom{N-k}{v} \hspace{-3pt} \left(e^{-\frac{\mu}{N} t} \omega \right)^v \hspace{-3pt} \left(e^{-\lambda_e t} (1-\omega) \right)^{N-k-v} \hspace{-5pt}  dt \notag \allowdisplaybreaks[1]  \\
			& = \sum_{j=0}^{s-1} \binom{N}{j} \sum_{k=0}^{j} \binom{j}{k}  (-1)^{N-k} \sum_{v=0}^{N-k} \binom{N-k}{v}   \omega^v (1-\omega)^{N-k-v} \notag\\ 					& \cdot \int_0^\infty \hspace{-3pt}   \left(e^{-\frac{\mu}{N} t}  \right)^v  \left(e^{-\lambda_e t} \right)^{N-k-v}   dt \notag \\
			&=  \sum_{j=0}^{s-1} \hspace{-3pt} \binom{N}{j} \hspace{-3pt} \sum_{k=0}^{j} \hspace{-3pt} \binom{j}{k} \hspace{-3pt}  (-1)^{N-k}  \sum_{v=0}^{N-k} \hspace{-3pt} \binom{N-k}{v}  \frac{\omega^v (1-\omega)^{N-k-v} }{\lambda_e(N-k-v)+\frac{\mu}{N} v},
\end{align} where \eqref{w0} follows from the derivations in  \eqref{lineW2} and \eqref{w2} follows from Lemma \ref{lemmaPW}, where we denoted $\omega=1-\nu(0) \frac{\frac{\mu}{N}}{\frac{\mu}{N} -\lambda_e}$. 
\end{proof}

\subsection{Asymptotic Analysis of Retrieval Time of Estimate}

In this section we determine $\mathbb{E}[W_s]$ for asymptotically large rate of energy harvesting, battery capacity and size of the sensor network. First, we introduce the following  lemma.
\begin{lemma}\label{technicallemma}
For any $0 \leq k \leq s$, $k \in \mathbb{N}$,
\begin{align} \label{sumLimit}
\binom{N}{k} \sum_{v=0}^k \binom{k}{v} (-1)^{k-v} \frac{N-k}{N-v} = 1.
\end{align}
\end{lemma}

\begin{proof}
This proof follows from induction on $k$.

It is easy to see that \eqref{sumLimit} holds for $k=0$. We assume that \eqref{sumLimit} holds for some $k>0$. We next show that  \eqref{sumLimit} holds for $k+1$.
\begin{align*}
&\binom{N}{k+1} \sum_{v=0}^{k+1} \binom{k+1}{v} (-1)^{k+1-v} \frac{N-(k+1)}{N-v}  \\ 
		&= \binom{N}{k+1} \sum_{v=0}^{k+1} \binom{k+1}{v} (-1)^{k+1-v} \frac{N-v+v-(k+1)}{N-v} \\
		&=  \binom{N}{k+1} \sum_{v=0}^{k+1} \binom{k+1}{v} (-1)^{k+1-v} (1)^v \\
		&+  \binom{N}{k+1} \sum_{v=0}^{k+1} \binom{k+1}{v} (-1)^{k+1-v} \frac{v-(k+1)}{N-v} \\
		&=0 + \binom{N}{k+1} \sum_{v=0}^{k+1} \binom{k+1}{v} (-1)^{k+1-v-1} \frac{(k+1)-v}{N-v} \allowdisplaybreaks[1] \\
		&=  \binom{N}{k+1} \sum_{v=0}^{k} \binom{k+1}{v} (-1)^{k-v} \frac{(k+1)-v}{N-v} \\
		&=  \binom{N}{k+1} (k+1) \sum_{v=0}^{k} \frac{k!}{v! (k+1-v-1)!} (-1)^{k-v} \frac{1}{N-v} \\
		&= \binom{N}{k}  \sum_{v=0}^{k} \binom{k}{v} (-1)^{k-v} \frac{N-k}{N-v} =1,
\end{align*} where the last equality follows from the induction hypothesis.

\end{proof}

\begin{theorem}\label{theoremLambdaInfinity}
For $1 \leq N < \infty$ and $0 < B < \infty$,
 \[
\lim_{\lambda_e \rightarrow \infty} \mathbb{E}[W_s]= \sum_{j=0}^{s-1} \frac{1}{\mu(1-j/N)}.
\]
\end{theorem}

\begin{proof}

Taking $\lambda_e \rightarrow \infty$ in Theorem \ref{firstmomentW}, we have that
\begin{align*}
\lim_{\lambda_e \rightarrow \infty} \mathbb{E}[W_s] 
				&=  \sum_{j=0}^{s-1} \binom{N}{j} \sum_{v=0}^{j} \binom{j}{v} (-1)^{j-v} \frac{1}{\mu (1-v/N)}.
\end{align*}
The result now follows from Lemma \ref{technicallemma}.
\end{proof}

We next consider the situations when the capacity of the sensors to store energy in the battery is unlimited, i.e. $B \rightarrow \infty$. 

For $\lambda_e < \mu/N$ and $B \rightarrow \infty$, the battery of a sensor is most of the time empty as the rate at which this sensor receives energy is lower than the rate at which this sensor broadcasts, and thus, consumes energy. As a consequence, in this case, the waiting time for a client to retrieve $s$ measurements from distinct sensors largely depends on $\lambda_e$, which supports the broadcasting process.

For $\lambda_e >\mu/N$ and $B \rightarrow \infty$, a sensor has most of the time energy for broadcasting since the rate at which it harvests new energy is higher than the rate at which it broadcasts. In this case, the waiting time for a client to retrieve $s$ measurements from distinct sensors depends on the broadcasting rate $\mu/N$. 

\begin{theorem}
For $1 \leq N < \infty$ and $0 < \lambda_e < \infty$,
 \[ 
\lim_{B \rightarrow \infty} \mathbb{E}[W_s]=
\begin{cases}
\displaystyle
\sum_{j=0}^{s-1} \frac{1}{\lambda_e(N-j)}, &\lambda_e<\frac{\mu}{N}\\
\displaystyle
\sum_{j=0}^{s-1} \frac{1}{\frac{\mu}{N}(N-j)}, & \lambda_e \geq \frac{\mu}{N}.
\end{cases}
\]
\end{theorem}

\begin{proof} 

We first consider the case $\lambda_e < \mu/N$. 
Then, from Lemma \ref{pisensor}, $\lim_{B \rightarrow \infty } \nu(0) = 1- \frac{\lambda_e}{\mu/N}$.  Using this result in Theorem \ref{firstmomentW} we have that
\[
\lim_{B \rightarrow \infty} \mathbb{E}[W_s]= \sum_{j=0}^{s-1} \binom{N}{j}   \sum_{k=0}^j \binom{j}{k} (-1)^{j-k}  \frac{1}{\lambda_e(N-k)}.
\]
The result follows from Lemma \ref{technicallemma}.

We next consider the case $\lambda_e \geq \mu/N$. Then $\nu(0) \rightarrow 0$. Using this result in Theorem \ref{firstmomentW} we have that
\[
\lim_{B \rightarrow \infty} \mathbb{E}[W_s]= \sum_{j=0}^{s-1} \binom{N}{j}   \sum_{k=0}^j \binom{j}{k} (-1)^{j-k}  \frac{1}{\mu(N-k)/N}.
\]
Again, the result follows from Lemma \ref{technicallemma}.

\end{proof}

\begin{theorem} For $0 < B < \infty$ and $0 < \lambda_e < \infty$, 
\[
\lim_{N \rightarrow \infty} \mathbb{E}[W_s] = \frac{s}{\mu}.
\]
\end{theorem}

\begin{proof}
Recall that a measurement is broadcasted at an exponential rate $\mu$ from the sensor network, for any $N>0$. Moreover, given that $N \rightarrow \infty$, the probability that there is always at least one sensor with energy which can transmit, tends to 
$1$. Also, the probability that any $s$ consecutive broadcasts are from distinct sensors, tends to 1 as $N \rightarrow \infty$. From the above arguments it follows that, for each of the $s$ measurements, a client waits, in expectation, $1/\mu$.  The result follows.
\end{proof}

%
%
%

\section{Numerical Results}\label{Numerical}

In this section we analyze numerically the expected waiting time for a client to retrieve a reliable estimate of an attribute under various assumptions concerning the size of the wireless sensor network, the maximum battery capacity of the sensors and the rate at which sensors harvest energy from the environment. 

Figures \ref{fig:lambda>mu-a} and \ref{fig:lambda<mu-b} show $\mathbb{E}[W_s]$ under various $N$, the size of the wireless sensor network. As expected, as $N$ increases, $\mathbb{E}[W_s]$ decreases. The reason is that, as $N$ increases, the probability that at least one sensor has battery to broadcast an innovative measurements increases. Thus, it is expected that clients wait less to retrieve $s$ measurements to compute a reliable estimate of an attribute. Moreover, Figures \ref{fig:lambda>mu-a} and \ref{fig:lambda<mu-b} show that, for a fixed $\lambda_e$, if $B$ is increased, then $\mathbb{E}[W_s]$ decreases. This is because as $B$ increases, more energy can be collected in the sensor batteries, which enables broadcasts.  

Figure \ref{fig:lambda>mu-a} considers the case when $\lambda_e \geq \mu/N$, whereas Figure  \ref{fig:lambda<mu-b} considers the case when $\lambda_e < \mu/N$. When $\lambda_e \geq \mu/N$, it is expected that most of the time the batteries of the sensors have energy. If  $\lambda_e < \mu/N$, the batteries are expected to be empty most of the time. This explains the fact that $\mathbb{E}[W_s]$ takes lower values in Figure \ref{fig:lambda>mu-a} than in Figure  \ref{fig:lambda<mu-b}.

\begin{figure}[!htb]
    \centering
    \begin{minipage}{.48\textwidth}
        \centering
        \includegraphics[scale=0.5]{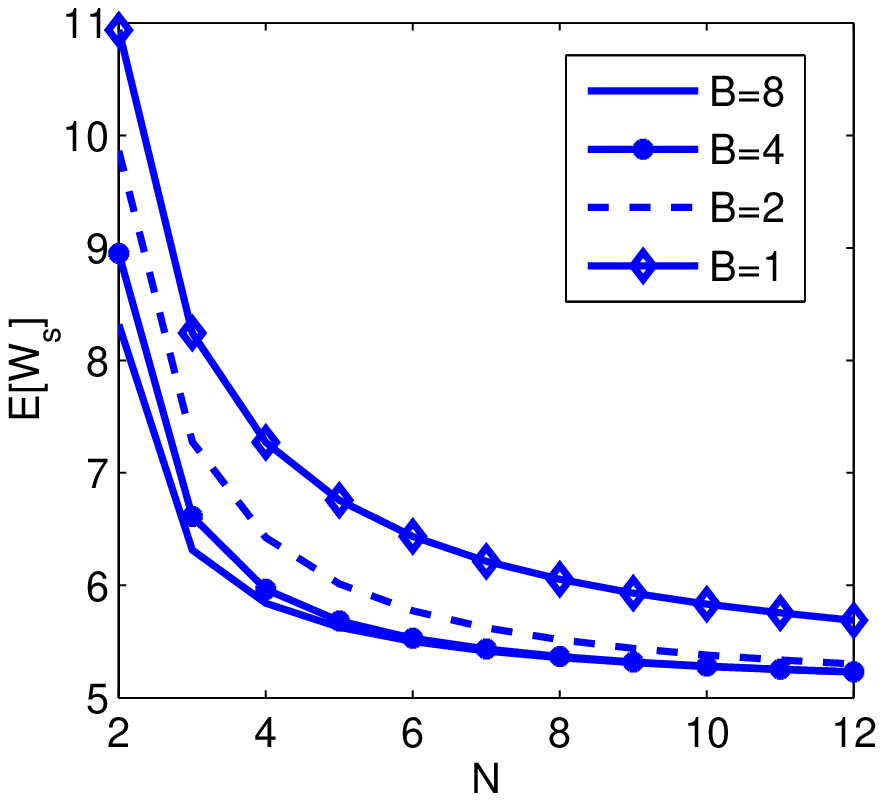}
        \caption{$\lambda_e=0.2$,  $ \mu=0.4, s=2$.}
        \label{fig:lambda>mu-a}
    \end{minipage}%
    \begin{minipage}{0.48\textwidth}
        \centering
        \includegraphics[scale=0.5]{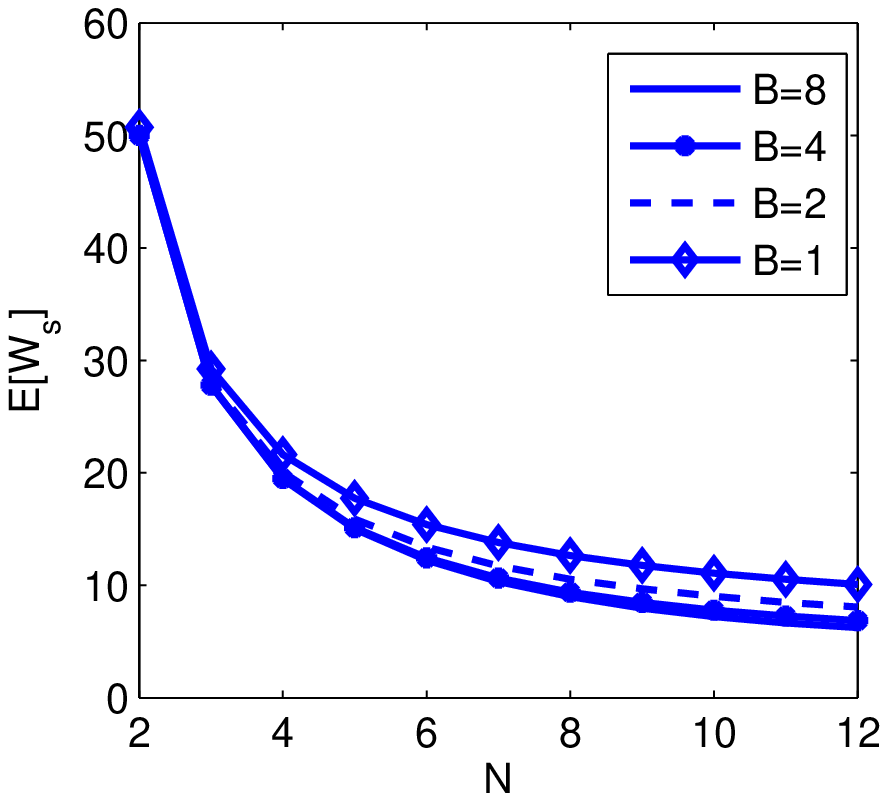}
        \caption{$\lambda_e=0.03$, $\mu=0.4, s=2$.}
        \label{fig:lambda<mu-b}
    \end{minipage}
\caption{$\mathbb{E}[W_s]$ under various $N$, the size of the sensor network.} \label{varyN}
\end{figure}

\begin{figure}[ht!]
    \centering
\includegraphics[scale=0.6]{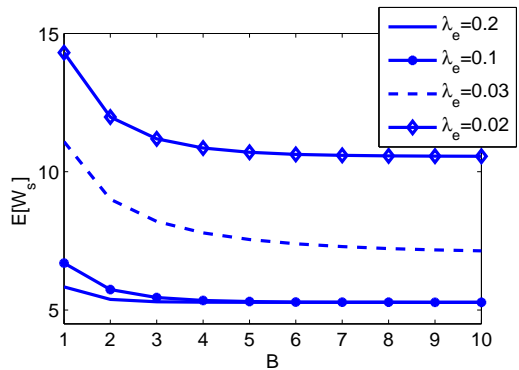}
    \caption{$\mathbb{E}[W_s]$ under various $B$, the maximum battery capacity of a sensor, $N=10$, $\mu=0.4, s=2$. }
    \label{varyB}
\end{figure}

Figure \ref{varyB} shows $\mathbb{E}[W_s]$ under various $B$, the battery capacity of a sensor, and for various $\lambda_e$, the rate at which sensors harvest energy from the environment. As expected, for a fixed $B$, $\mathbb{E}[W_s]$ decreases as $\lambda_e$ increases. This is the case because the battery of the sensors are more frequently replenished and, thus, the sensors have energy to broadcasts their measurements. 
We note that for $\lambda_e \in \{0.1, 0.2\}$, $\lambda_e > \mu/N$, while for $\lambda_e \in \{0.03, 0.02\}$,  $\lambda_e < \mu/N$.
Figure \ref{varyB} also shows that, for a fixed $\lambda_e$, if $B$ increases, then $\mathbb{E}[W_s]$ decreases. This decrease becomes less significant for large values of $B$. 
This can be explained as follows. In the case that $\lambda_e \geq \mu/N$, even though sensors are able to store large amounts of energy, i.e. $B$ is large, the rate at which the sensors broadcast is low and thus, $\mathbb{E}[W_s]$ mostly depends on the broadcasting rate, rather than $B$. In the case that $\lambda_e < \mu/N$, even though $B$ is large, the amount of energy in the batteries is expected to be low most of the times. Thus, in this case, the fact that $B$ is very large does not result in a significantly decrease in $\mathbb{E}[W_s]$.

\section{Conclusions}\label{Conclusions}

In this paper, we considered the problem of retrieving a reliable estimate of an attribute from a wireless sensor network that harvests energy from the environment. We assumed that energy is available for harvesting at each sensor according to a Poisson process. Moreover, the sensors store the harvested energy in a limited-capacity battery. Provided there is sufficient energy stored, the sensors broadcast their measurements, independently, at random. 

We showed that  the time until an arbitrary sensor broadcasts has a phase-type distribution. Based on this, we determined the probability distribution of the time to retrieve from the sensor network a reliable estimate of an attribute. We also provided a closed-form expression for the expected time to retrieve this estimate. In addition, we determined the retrieval time of a reliable estimate when the energy available for harvesting, the storage capacity of the sensor battery or the size of the sensor networks is asymptotically large. 

Lastly, we analyzed numerically the retrieval time of a reliable estimate under various assumptions concerning the size of the wireless sensor network, the maximum capacity of the batteries of sensors, as well as the rate at which sensors harvest energy from the environment. The numerical results show that deploying sensors with very large  batteries does not result in a significant decrease in the retrieval time of a reliable estimate. On the other hand, the numerical results show that the rate at which sensors harvest energy and the rate at which they broadcast significantly influences the retrieval time of a reliable estimate.

\bibliographystyle{spmpsci}
\bibliography{IEEEabrv,mybibfile}

\begin{thebibliography}{10}
\providecommand{\url}[1]{{#1}}
\providecommand{\urlprefix}{URL }
\expandafter\ifx\csname urlstyle\endcsname\relax
  \providecommand{\doi}[1]{DOI~\discretionary{}{}{}#1}\else
  \providecommand{\doi}{DOI~\discretionary{}{}{}\begingroup
  \urlstyle{rm}\Url}\fi

\bibitem{ahsanullah2013introduction}
Ahsanullah, M., Nevzorov, V.B., Shakil, M.: An introduction to order
  statistics.
\newblock Springer (2013)

\bibitem{akkouchi2008convolution}
Akkouchi, M.: On the convolution of exponential distributions.
\newblock J. Chungcheong Math. Soc \textbf{21}(4), 501--510 (2008)

\bibitem{akyildiz2002wireless}
Akyildiz, I., Su, W., Sankarasubramaniam, Y., Cayirci, E.: Wireless sensor
  networks: A survey.
\newblock Computer networks \textbf{38}(4), 393--422

\bibitem{antepli2011optimal}
Antepli, M.A., Uysal-Biyikoglu, E., Erkal, H.: Optimal packet scheduling on an
  energy harvesting broadcast link.
\newblock IEEE Selected Areas in Communications \textbf{29}(8), 1721--1731
  (2011)

\bibitem{asmussen2008applied}
Asmussen, S.: Applied probability and queues, vol.~51.
\newblock Springer Science \& Business Media (2008)

\bibitem{devillers2012general}
Devillers, B., Gunduz, D.: A general framework for the optimization of energy
  harvesting communication systems with battery imperfections.
\newblock IEEE Communications and Networks \textbf{14}(2), 130--139 (2012)

\bibitem{ho2012optimal}
Ho, C.K., Zhang, R.: Optimal energy allocation for wireless communications with
  energy harvesting constraints.
\newblock IEEE Transactions on Signal Processing \textbf{60}(9), 4808--4818
  (2012)

\bibitem{kuan2014reliable}
Kuan, C.C., Lin, G.Y., Wei, H.Y., Vannithamby, R.: Reliable multicast and
  broadcast mechanisms for energy-harvesting devices.
\newblock IEEE Transactions on Vehicular Technology \textbf{63}(4), 1813--1826
  (2014)

\bibitem{lei2009generic}
Lei, J., Yates, R., Greenstein, L.: {A} generic {M}odel for {O}ptimizing
  {S}ingle-{H}op {t}ransmission {P}olicy of {R}eplenishable {S}ensors.
\newblock IEEE Transactions on Wireless Communications \textbf{8}(2), 547--551
  (2009)

\bibitem{nielsen1988modelling}
Nielsen, B.F.: Modelling of multiple access systems with phase type
  distributions.
\newblock Imsor (1988)

\bibitem{ozel2011transmission}
Ozel, O., Tutuncuoglu, K., Yang, J., Ulukus, S., Yener, A.: Transmission with
  energy harvesting nodes in fading wireless channels: Optimal policies.
\newblock IEEE Journal on Selected Areas in Communications \textbf{29}(8),
  1732--1743 (2011)

\bibitem{paradiso2005energy}
Paradiso, J.A., Starner, T.: Energy scavenging for mobile and wireless
  electronics.
\newblock IEEE Pervasive Computing \textbf{4}(1), 18--27 (2005)

\bibitem{sharma2010optimal}
Sharma, V., Mukherji, U., Joseph, V., Gupta, S.: Optimal energy management
  policies for energy harvesting sensor nodes.
\newblock Wireless Communications, IEEE Transactions on \textbf{9}(4),
  1326--1336 (2010)

\bibitem{sudevalayam2011energy}
Sudevalayam, S., Kulkarni, P.: Energy harvesting sensor nodes: Survey and
  implications.
\newblock IEEE Communications Surveys \& Tutorials \textbf{13}(3), 443--461
  (2011)

\bibitem{tandon2014has}
Tandon, A., Motani, M.: Has green energy arrived? delay analysis for energy
  harvesting communication systems.
\newblock In: Proceedings of IEEE International Conference on Sensing,
  Communication, and Networking, pp. 582--590 (2014)

\bibitem{tutuncuoglu2012optimum}
Tutuncuoglu, K., Yener, A.: Optimum transmission policies for battery limited
  energy harvesting nodes.
\newblock IEEE Transactions on Wireless Communications, \textbf{11}(3),
  1180--1189 (2012)

\bibitem{yang2012broadcasting}
Yang, J., Ozel, O., Ulukus, S.: Broadcasting with an energy harvesting
  rechargeable transmitter.
\newblock IEEE Transactions on Wireless Communications \textbf{11}(2), 571--583
  (2012)

\bibitem{yang2012optimal}
Yang, J., Ulukus, S.: Optimal packet scheduling in an energy harvesting
  communication system.
\newblock IEEE Transactions on Communications \textbf{60}(1), 220--230 (2012)

\bibitem{zhang2011computing}
Zhang, X., Hou, Z.: Computing the moments of order statistics from
  nonidentically distributed phase-type random variables.
\newblock Journal of computational and applied mathematics \textbf{235}(9),
  2897--2903 (2011)

\end{thebibliography}

\end{document}